\documentclass[article,longbibliography]{revtex4-2}

\usepackage [utf8]{inputenc}
\usepackage{ulem}

\usepackage[
    colorlinks=true,
    linkcolor=Teal,
    filecolor=Teal,
    citecolor = Crimson,      
    urlcolor=Crimson]{hyperref}
\usepackage{mathptmx}
\usepackage{graphicx}
\usepackage{ae,aecompl}
\usepackage{amssymb}
\usepackage{bm}
\usepackage{extarrows}
\usepackage{tikz-cd}
\usetikzlibrary{matrix}
\usepackage{mathtools}
\usepackage[ruled,vlined,norelsize]{algorithm2e}
\usepackage{diagbox}
\usepackage{mdframed}
 
\newtheorem{theo}{Theorem}
\newenvironment{ftheo}
  {\begin{mdframed}\begin{theo}}
  {\end{theo}\end{mdframed}}

\usepackage{xcolor}

\newcommand{\ket}[1]{\left|#1\right\rangle}
\newcommand{\bra}[1]{\left\langle#1\right|}

\usepackage{paralist}
\definecolor{Teal}{rgb}{0.0, 0.48, 0.46}
\definecolor{Crimson}{RGB}{220, 20, 60}

\newtheorem{corollary}{Corollary}
\newtheorem{theorem}{Theorem}
\newtheorem{proposition}{Proposition}

\newtheorem{defi}{Definition}

\newcommand{\tr}{\text{tr}}

\newcommand{\dd}{\mathfrak{d}}

\newenvironment{proof}{\paragraph{Proof:}}{\hfill$\square$}

\begin{document}
\title{No-go theorems for photon state transformations in quantum linear optics}

\date{\today}
\author{Pablo V. Parellada}
\email{pablo.veganzones@uva.es}  
\affiliation{Universidad de Valladolid, Dpto. Teor\'ia de la Se\~{n}al e Ing. Telem\'atica, Paseo Bel\'en n$^o$ 15, 47011 Valladolid, Spain}

\author{Vicent Gimeno i Garcia}
\email{gimenov@uji.es}
\affiliation{Department of Mathematics, Universitat Jaume I-IMAC,   E-12071, 
Castell\'{o}, Spain}

\author{Julio Jos\'e Moyano-Fern\'andez}
\email{moyano@uji.es}  
\affiliation{Department of Mathematics, Universitat Jaume I-IMAC,   E-12071, 
Castell\'{o}, Spain}

\author{Juan Carlos Garcia-Escartin}
\email{juagar@tel.uva.es}  
\affiliation{Universidad de Valladolid, Departamento de Teor\'ia de la Se\~{n}al e Ing. Telem\'atica, Paseo Bel\'en n$^o$ 15, 47011 Valladolid, Spain}

\begin{abstract}
We give a necessary condition for photon state transformations in linear optical setups preserving the total number of photons. From an analysis of the algebra describing the quantum evolution, we find a conserved quantity that appears in all allowed optical transformations. We give some examples and numerical applications, with example code, and give three general no-go results. These include (i) the impossibility of deterministic transformations which redistribute the photons from one to two different modes, (ii) a proof that it is impossible to generate a perfect Bell state with an arbitrary ancilla from the Fock basis and (iii) a restriction for the conversion between different types of entanglement (converting GHZ to W states). These tools and results can help in the design of experiments for optical quantum state generation.

\end{abstract}

\maketitle




\section{Overview and summary of the results}\label{descriptiom}
Linear optical systems acting on $n$ photons in $m$ different modes offer a rich evolution and their study is not only interesting in its own right, but has applications to linear optics quantum computing \cite{KLM01,KMN07} and in boson sampling, a proposed demonstration that quantum systems can perform tasks impossible for any classical device \cite{AA11,Aar11}. 

These simple optical setups are sometimes called linear optics multiports or linear interferometers and appear in many experiments in both classical and quantum optics. 

In this work we study the transformations of photon number states of the form $\ket{n_1\ldots n_m}$ in a system with $m$ modes where mode $i$ has $n_i$ photons. We give a necessary condition for an allowed transformation
\begin{equation}
\ket{\psi_{\text{in}}}=\sum_{i=1}^{M}\alpha_i \ket{i} \to
\ket{\psi_{\text{out}}}=\sum_{i=1}^{M}\beta_i \ket{i}
\end{equation}
for general superpositions of all the possible states of the form  $\ket{i}=\ket{n_1^{i}\ldots n_m^i}$ with $\sum_{j=1}^{M}n_j^{i}=n$. The condition can be also used for mixed states.

The criterion shows the limitations of linear optical interferometers and can be used to derive a few no-go results. For instance, we show that it is impossible to generate a Bell state deterministically from a separable photon number input in heralded schemes with any number of ancillary modes and photons.

These results come from an analysis of the density matrix of the input and output states in terms of the algebra that gives the effective Hamiltonians of the evolution in linear optical systems. The adjoint representation of the Lie groups and algebras that describe the problem have quantities that are conserved under linear optical evolution. We give an explicit formulation of the resulting necessary condition for a valid transformation, which can be used to identify impossible transitions from a given input to a target optical state. The examples in this paper, as well as general functions to compute the conserved quantities, are available for Python as an extension to the QOptCraft open software \cite{QOptCraft}.

First we introduce the theory of quantum linear optical interferometers (Section  \ref{intro}) and their description in terms of Lie algebras (Section \ref{imagealgebra}). Then, in Section \ref{adjoint}, we find an invariant associated to valid optical transformations with linear optics which results from an analysis of the adjoint map. Section \ref{forbidden} gives the main result of the paper: a necessary condition for allowed transformations between photon number states in linear interferometers. Section \ref{formulae} presents some closed formulae for the criterion. We conclude by giving three consequences of the necessary condition. In Section \ref{ntonminusk} we show it restricts the possible transformation between states in the Fock basis. Section \ref{bell} gives an impossibility result for the generation of dual-rail Bell states from separable photon number states using linear optics, even for an arbitrary number of ancillary photons and modes. Section \ref{GHZtoW} has a similar proof for the impossibility of converting between different types of entanglement for GHZ and W states. We close the paper with an overview of the results in Section \ref{conclusions}.

\section{Quantum evolution in linear optical systems}\label{intro}
Classically, a linear interferometer acting on $m$ modes (or ports) is completely described by an $m\times m$ unitary scattering matrix $S$ \cite{Hau95,Poz04}. 

In quantum optics, a basic problem is determining the effect of a linear interferometer on number states $\ket{n_1,\ldots , n_m}$ where the $i$th mode has $n_i$ photons and $\sum_i^m n_i=n$. Each photon number state and each mode is orthogonal to the rest. The relevant Hilbert space for all the possible states, $\mathcal{H}_{m,n}$, is isomorphic to $\mathbb{C}^M$ with $M=\dim_{\mathbb{C}} \mathcal{H}_{m,n}= \binom{m+n-1}{n}$.  

For a fixed number of photons and modes, $M$ corresponds to all the possible ways of distributing the $n$ indistinguishable photons into the $m$ modes. The dimension of the Hilbert space corresponds to the ways of placing $n$ balls into $m$ boxes and it is sometimes called the number of weak compositions of $n$ into $m$ (all the combinations of $k$ nonnegative numbers that sum to $n$) or multiset coefficients.
  
The quantum evolution in $\mathcal{H}_{m,n}$ is described by an $M\times M$ unitary matrix $U$ which is derived from $S$ using a photonic homomorphism $\varphi_{n,m}(S)$ that goes from the unitary group $U(m)$ into $U(M)$. There are different equivalent definitions of this homomorphism and its properties \cite{Cai53,SGL04,Sch04,AA11}.

For a linear $m$-port, $S$ has only $2m^2-1$ degrees of freedom, while, in order to generate any possible $M\times M$ matrix $U$ we need to be able to set $2M^2-1$ degrees of freedom. Except for trivial cases when $n=1$ or $m=1$, linear optics cannot produce all the quantum operations which are possible in $\mathcal{H}_{m,n}$ \cite{MG17}. When $m=1$, the multiport is reduced to a phase shifter $e^{i\phi}$ and all the possible operations are trivial. When $n=1$, $U=\varphi(S)=S$ and then any desired quantum evolution is possible (but this is not very useful as the number of modes would equal the dimension of the Hilbert space). Once we know $S$, there exist multiple methods that describe explicit implementations using only beam splitters and phase shifters \cite{RZB94,BA14,Saw16,CHM16,GMS18}.

For any given $U$ and fixed $n$ and $m$, it is possible to determine if there is any linear system which can provide the desired evolution and, if that is the case, find a suitable scattering matrix $S$ \cite{GGM19}.

In this paper, we study a different problem: whether an input state $\ket{\psi_{\text{in}}}=\sum_{i=1}^{M}\alpha_i \ket{i}$ for all the $\ket{i}$ possible states in the form $\ket{i}=\ket{n_1^{i}\ldots n_m^i}$ with $\sum_{j=1}^{M}n_j^{i}=n$ can be taken to an output state $\ket{\psi_{\text{out}}}=\sum_{i=1}^{M}\beta_i \ket{i}$.

\section{Preliminaries: The image algebra}\label{imagealgebra}
In certain applications, the linear multiport is best described in terms of the Hermitian matrices giving the effective Hamiltonians that provide an alternative description of the evolutions $S$ and $U$ \cite{LN04,GGM18}. We can define two effective Hamiltonians $H_S$ and $H_U$ so that we can recover the unitary matrices through exponentiation with $S=e^{iH_s}$ and $U=e^{iH_U}$. The skew-Hermitian matrices $iH_S \in \mathfrak{u}(m)$ and $iH_U \in \mathfrak{u}(M)$ live in the algebras $\mathfrak{u}(m)$ and $\mathfrak{u}(M)$ associated to the unitary groups of the scattering matrix $S$ and the full quantum evolution $U$ in the larger Hilbert space with $n$ photons. 

We are limited to the image of the photonic homomorphism $\varphi$. Any unitary that can be realized with lossless linear interferometers preserving the total photon number $U \in \mathrm{im}{\varphi}$ can be described as the exponential of a skew-Hermitian matrix $i H_U\in \mathrm{im}{d\varphi}$ from the image subalgebra of $\mathfrak{u}(m)$ which is smaller than $\mathfrak{u}(M)$ in systems with more than one photon or one mode.

\subsection{Optical realizations}\label{opticalrealizations}

Any linear interferometer with $m$ modes is completely described by an $m\times m$ unitary scattering matrix $S$ which has a limited number of degrees of freedom, $2m^2-1$. This means that, except for trivial cases when $n=1$ or $m=1$, linear multiports can only provide a limited subset of all the possible operations over $n$ photons in $m$ modes, which are described by $M\times M$ unitary matrices $U$ with $2M^2-1$ degrees of freedom \cite{MG17}. 

The scattering matrices $S$ are elements of the unitary group $U(m)$ and any general evolution $U$ on $\mathcal{H}_{m,n}$ is an element of the unitary group $U(M)$. The subgroup of all the operations which can be implemented with linear optics is described by the image subgroup of $\varphi_{m,M}$, namely $\mathrm{im} \varphi_{m,M} = \{B \in U(M) : \exists \, A \in U(m) ~\mbox{such that}~\varphi_{m,M}(A)=B\}$.

We say 
\begin{defi}
A matrix $U \in U(M)$ is an $(m,n)$-optical realization if $U \in \mathrm{im} \varphi_{m,M}$.
\end{defi}

The map $\varphi_{m,M}$ is a differentiable group homomorphism \cite{AA11} inducing an algebra homomorphism, $d \varphi_{m,M}$. The commutative diagram 
\begin{center}
	\begin{tikzcd}
		\mathfrak{u}(m) \arrow{r}{d \varphi_{m,M}}\arrow{d}[swap]{\exp} & \mathfrak{u}(M)\arrow{d}{\exp} \\
		U(m)\arrow{r}[swap]{ \varphi_{m,M}} & U(M) \\
	\end{tikzcd}
\end{center}
relates the unitary groups $U(m)$ and $U(M)$ containing the scattering matrix $S$ and the $n$-photon evolution operator $U$, respectively, to the algebras $\mathfrak{u}(m)$ and $\mathfrak{u}(M)$ whose elements correspond to skew-Hermitian matrices $iH_S$ and $iH_U$ which give an equivalent description of the evolution through exponentiation of the Hamiltonians $H_S$ and $H_U$ \cite{LN04,GGM18}.

The image algebra can be described from the differential $d \varphi_{m,M}(iH_S) =\sum_{jk}iH_{Sjk}\hat{a}_j^{\dag} \hat{a}_k$ for $iH_S\in \mathfrak{u}(m)$ (see \cite{GGM18} for the details), where $\hat{a}_i^{\dag}$ and $\hat{a}_i$ are the photon creation and annihilation operators for mode $i$ respectively. 

\subsection{The image algebra}
\label{Obasis}
Consider the canonical basis of $\mathfrak{u}(m)$ $\{\ket{1}=\ket{1,0,\ldots , 0},\ket{2}=\ket{0,1,\ldots, 0}, \ldots , \ket{m}=\ket{0,\ldots, 0, 1}  \}$. The matrices
\begin{align}
\label{basisdef}
e_{jk}:=& \frac{i}{2}\big (\ket{j}\bra{k}+\ket{k}\bra{j}\big ) \mbox{ for } \, k\leq j=1,\ldots , m; \\
f_{jk}:=& \frac{1}{2} \big (\ket{j}\bra{k}-\ket{k}\bra{j}\big )  \mbox{ for } \, k< j=1,\ldots , m, \nonumber
\end{align} 
for a basis of $\mathfrak{u}(m)$ and the matrices
\begin{align*}
d \varphi_{m,M}(e_{jk})=& \frac{i}{2}\big (\hat{a}^\dag_j \hat{a}_k+\hat{a}^\dag_k \hat{a}_j\big ) \neq 0,\\
d \varphi_{m,M}(f_{jk})=& \frac{1}{2} \big (\hat{a}^\dag_j \hat{a}_k-\hat{a}^\dag_k \hat{a}_j\big ) \neq 0,
\end{align*}
give a basis for the image algebra $\mathfrak{d}=\mathrm{im} \hspace{0.5ex}d\varphi_{m,M} \subset \mathfrak{u}(M)$ (the map $d \varphi_{m,M}$ is injective). 
We can label the $m^2$ elements of the basis of the image algebra as: 

\begin{align}
\label{ImBasis}
b_{j}^{n}&=i\hat{n}_j \mbox{ for } \, j=1,\ldots , m;\nonumber \\
b_{jk}^{e}&=\frac{i}{2}\big (\hat{a}^\dag_j \hat{a}_k+\hat{a}^\dag_k \hat{a}_j\big ) \mbox{ for } \, k< j=1,\ldots , m; \nonumber  \\
b_{jk}^{f}&=\frac{1}{2}\big (\hat{a}^\dag_j \hat{a}_k-\hat{a}^\dag_k \hat{a}_j\big ) \mbox{ for } \, k< j=1,\ldots , m. 
\end{align}

We consider projection into each basis element using the real inner product, a positive definite symmetric bilinear form,
\begin{equation}
\langle u,v \rangle:=\frac{1}{2}\mathrm{tr}(u^{\dag} v + v^{\dag}u).
\end{equation}

For our skew-Hermitian matrices $\langle u,v \rangle=-\mathrm{tr}(uv)$ and the product is symmetric as $\langle u,v \rangle = -\mathrm{tr}(uv) = -\mathrm{tr}(vu) = \langle v,u \rangle$ and is equivalent to the Hilbert-Schmidt inner product. This inner product induces a positive definite norm:
\begin{equation}
\langle u,u \rangle = \mathrm{tr}(u^{\dag}u)= ||u||^2.
\end{equation}
This gives a metric that is Riemannian and bi-invariant.  
\subsubsection{Orthonormalization. A simple example.}\label{orthexample}
The basis of $\mathrm{im}\, d\varphi_{m,M}$ given in Equation (\ref{ImBasis}) is not orthonormal, even if it was derived from the orthonormal canonical basis of $\mathfrak{u}(m)$. We can see an example for a linear interferometer with two photons and two modes ($m=n=2$). Take the state basis $\{\ket{20},\ket{11},\ket{02}\}$ with
\begin{equation}
\ket{20}=\left( \begin{array}{c}
1 \\
0\\
0
 \end{array} \right), \hspace{2ex}
\ket{11}=\left( \begin{array}{c}
0 \\
1\\
0
 \end{array} \right), \hspace{2ex}
\ket{02}=\left( \begin{array}{c}
0 \\
0\\
1
 \end{array} \right).
\end{equation}
For this state order, the basis of the image algebra can be written as the matrices
\begin{align}
b_1=&b_{1,2}^e=\frac{i}{2}(\hat{a}_1^\dag\hat{a}_2+\hat{a}_2^\dag\hat{a}_1)=\frac{i}{2}\left( \begin{array}{ccc}
0 & \sqrt{2} & 0 \\
\sqrt{2} & 0 & \sqrt{2}\\
0 & \sqrt{2} & 0
 \end{array} \right), \hspace{2ex}
b_2=b_{1}^n=i\hat{n}_1=i\left( \begin{array}{ccc}
2 & 0 & 0 \\
0 & 1 & 0\\
0 & 0 & 0
 \end{array} \right),  \nonumber \\
b_3=&b_{2}^n=i\hat{n}_2= i \left( \begin{array}{ccc}
0 & 0 & 0 \\
0 & 1 & 0\\
0 & 0 & 2
 \end{array} \right),\hspace{2ex}
b_4=b_{1,2}^f=\frac{1}{2}(\hat{a}_1^\dag\hat{a}_2-\hat{a}_2^\dag\hat{a}_1)=\frac{1}{2}\left( \begin{array}{ccc}
0 & -\sqrt{2} & 0 \\
\sqrt{2} & 0 & -\sqrt{2}\\
0 & \sqrt{2} & 0
 \end{array} \right).
\end{align}
The elements are determined by the chosen state order. The first element of $b_1$ (first column and first row) is computed from $\bra{20}b_1\ket{20}= \bra{20}\frac{i}{2} \sqrt{2}\ket{11}=0$. $\bra{11}b_1\ket{20}= \bra{11} \frac{i}{2} \sqrt{2} \ket{11}=\frac{i\sqrt{2}}{2}$ gives the second column of the first row. Repeating this procedure we can get all the given matrices.

This is a complete basis for the image algebra, but unlike the basis of $\mathfrak{u}(m)$ defined in (\ref{basisdef}), it is not orthonormal. We can quickly see $\langle b_2,b_2 \rangle=5\neq 1$ and $ \langle b_2,b_3 \rangle=1 \neq 0$.

We can orthonormalize the basis using the Gram-Schmidt process to get a series of orthonormal matrices $c_i$ with $i=1,\ldots, m^2$. In the calculations of the following sections, an orthonormal basis helps when computing the projections of elements of the larger algebra into the image algebra without counting twice the contributions of parts of the original skew-Hermitian matrix. 

\section{The adjoint map: invariants in linear optical transformations.}\label{adjoint}
Let $\mathrm{Ad}_U:\mathfrak{u}(M)\to \mathfrak{u}(M)$ be the adjoint map defined by $\mathrm{Ad}_U(v)=UvU^{\dag}$. We denote by $\mathfrak{d}$ the subalgebra $d \varphi_{m,M}(\mathfrak{u}(m))\subseteq \mathfrak{u}(M)$ . Notice that $d \varphi_{m,M}:\mathfrak{u}(m)\to \dd$ is a bijection.
 
A crucial result, proved in \cite{GGM19}, is that only the unitaries which can be obtained with linear optics define an adjoint map that takes the elements of $\dd$ into $\dd$. 

\begin{theorem}\label{teo:1}
	$U \in \mathrm{im} \varphi_{m,M} \Longleftrightarrow \mathrm{Ad}_U\mid_{\mathfrak{d}}$ is an automorphism.
\end{theorem}

This gives a way to determine whether a concrete $U\in U(M)$ is an element of $ \mathrm{im} \varphi_{m,M}$ or not \cite{GGM19}.

\subsubsection{Conserved quantities}
The inner product in $\mathfrak{u}(M)$ allows the decomposition of $\mathfrak{u}(M)$ in the orthogonal subspaces
$$
\mathfrak{u}(M)= \mathfrak{d}\oplus \mathfrak{d}_\perp.
$$
The subspace $\mathfrak{d}_\perp$ contains all the matrices in $\mathfrak{u}(M)$ orthogonal to the elements of the image subalgebra $\dd$ with respect to our inner product.
Thence, for any skew-Hermitian matrix $v\in \mathfrak{u}(M)$ we can define a unique decomposition 
\begin{equation}
v=v_{t}+v_{p},
\end{equation}
where $v_t \in \mathfrak{d}$ is called the tangent component of $v$ and $v_p\in \mathfrak{d}_\perp$ the perpendicular component. 

We can choose two orthonormal bases $\{c_i^{t}\}$ and $\{c_j^{p}\}$ for $\dd$ and $\dd_\perp$. Section \ref{Obasis} gives a method to create a basis $\{c_i^{t}\}$ with indices $i=1,\ldots,m^2$. We can then complete a basis for $\mathfrak{u}(M)$ and take the remaining elements $\{c_j^{p}\}$ with $j=1,\ldots,M^2-m^2$ for the basis of the orthogonal complement of $\dd$. 

With these two bases we have a complete description of $v$ in terms of projections to the tangent and perpendicular subspaces with
\begin{equation}
v=\sum_{i=1}^{m^2} \alpha_i^t c_i^t+\sum_{j=1}^{M^2-m^2} \alpha_j^t c_j^p
\end{equation}
for $\alpha_i^t=\langle c_i^t,v \rangle$ and  $\alpha_j^p=\langle c_j^p,v \rangle$. 

For the adjoint map under any $U\in  \mathrm{im} \varphi_{m,M} $ the projection into the tangent and perpendicular subspaces are preserved as can be shown in the following proposition:
\begin{proposition}\label{prop:1}The following statements are equivalent:
\begin{enumerate}
    \item[(1)] $U \in \mathrm{im} \varphi_{m,M}$.
    \item[(2)] For any $v\in\mathfrak{u}(M)$, $\mathrm{Ad}_U(v_{t})=Uv_{t}U^{\dag}\in \dd$.
    \item[(3)] For any $v\in\mathfrak{u}(M)$, $\mathrm{Ad}_U(v_{p})=Uv_{p}U^{\dag}\in \dd_\perp$.
\end{enumerate}
\end{proposition}
\begin{proof}
    The proof of the equivalence (1) $\Longleftrightarrow$ (2) is done by Theorem \ref{teo:1}. Observe that, if we assume $U \in \mathrm{im} \varphi_{m,M}$, then $ \mathrm{Ad}_U(v_{t})\in \dd$ since $v_t\in\dd$ by Theorem \ref{teo:1}. On the other hand, if $\mathrm{Ad}_U(v_{t})\in \dd$ for any $v\in\mathfrak{u}(M)$, then for any $w\in \dd$ we have that $\mathrm{Ad}_U(w)\in \dd$, since $w_t=w$, and by Theorem \ref{teo:1} we conclude $U \in \mathrm{im} \varphi_{m,M}$.

    To prove the equivalence  (2) $\Longleftrightarrow$ (3) we only have to observe the following relation:
    $$
    \begin{aligned}
    \langle \mathrm{Ad}_U(v_t),\, \mathrm{Ad}_U(v_p)\rangle =&\frac{1}{2}\mathrm{tr}\left(\left(Uv_tU^\dag\right)^\dag Uv_pU^\dag+\left(Uv_pU^\dag\right)^\dag Uv_tU^\dag\right)=\frac{1}{2}\mathrm{tr}\left(U\left(v_t^\dag v_p+v_p^\dag v_t\right)U^\dag\right)\\
    =&\frac{1}{2}\mathrm{tr}\left(v_t^\dag v_p+v_p^\dag v_t\right)=\langle v_t\, v_p\rangle=0.
    \end{aligned}
    $$
    Then, $\mathrm{Ad}_U(v_t)\in\dd$ if and only if $\mathrm{Ad}_U(v_p)\in\dd_\perp$.
\end{proof}
\begin{corollary}\label{coro:1}
    Let  $U \in \mathrm{im} \varphi_{m,M}$.  Let $v$ be an element of $\mathfrak{u}(M)$. Denote by $w=UvU^\dag$. Then: 
    \begin{enumerate}
    \item $\Vert v_t\Vert=\Vert w_t\Vert$,
    \item $\Vert v_p\Vert=\Vert w_p\Vert$.
    \end{enumerate}
\end{corollary}
\begin{proof} The definition of $w$ and the decomposition of $\mathfrak{u}(M)=\dd\oplus\dd_\perp$ yield
$$
w=w_t+w_p=UvU^\dag =U(v_t+v_p)U^\dag=Uv_tU^\dag+Uv_pU^\dag.
$$
Since the decomposition of $w$ is unique and Proposition \ref{prop:1} implies that $Uv_tU^\dag\in \mathfrak{d}$ and $Uv_pU^\dag \in \mathfrak{d}_\perp$, we conclude that
$$
w_t=Uv_tU^\dag,\quad w_p=Uv_pU^\dag.
$$
Therefore
$$
\Vert w_t \Vert^2=\Vert Uv_tU^\dag\Vert =\frac{1}{2}\mathrm{tr}\left( \left(Uv_tU^\dag\right)^\dag Uv_tU^\dag\right)=\mathrm{tr}\left(Uv_t^\dag v_t U^\dag\right)= \mathrm{tr}\left(v_t^\dag v_t \right)=\langle v_t,\, v_t\rangle=\Vert v_t\Vert^2,
$$
and 
$$
\Vert w_p \Vert^2=\Vert Uv_pU^\dag\Vert =\frac{1}{2}\mathrm{tr}\left( \left(Uv_pU^\dag\right)^\dag Uv_pU^\dag\right)=\mathrm{tr}\left(Uv_p^\dag v_p U^\dag\right)= \mathrm{tr}\left(v_p^\dag v_p \right)=\langle v_p,\, v_p\rangle=\Vert v_p\Vert^2.
$$
\end{proof}

Corollary \ref{coro:1} leads us to the definition of a tangent invariant
\begin{equation}
\label{It}
I_t(v):=\Vert v_t\Vert^2=\sum_{i=1}^{m^2} |\langle v, c_i^t \rangle|^2,
\end{equation}
and a perpendicular invariant
\begin{equation}
\label{Ip}
I_p(v):=\Vert v_p\Vert^2=\sum_{j=1}^{M^2-m^2} |\langle v, c_j^p \rangle|^2,
\end{equation}
which are conserved under the adjoint map for a $U\in  \mathrm{im} \varphi_{m,M} $ in the following sense:
\begin{corollary}
    If $U\in  \mathrm{im} \varphi_{m,M}$ and $v\in\mathfrak{u}(M)$, then
    \begin{enumerate}
        \item The tangent invariant is conserved under adjoint transformations as
        $$
I_t(\mathrm{Ad}_U(v))=I_t(v).
        $$
        \item Similarly, the perpendicular invariant is conserved under adjoint transformations as
        $$
I_p(\mathrm{Ad}_U(v))=I_p(v).
        $$
    \end{enumerate}
\end{corollary}

These quantities remind of the ``energies'' in Fourier series expansions and will be conserved under linear optical transformations $U \in \mathrm{im} \varphi_{m,M}$.

\section{Allowed and forbidden transformations}\label{forbidden}
The existence of conserved quantities gives us a necessary criterion that identifies when a particular input photonic state cannot be taken to a desired output state using only linear optics. We just need to describe the evolution in terms of skew-Hermitian matrices.

\subsection{Density matrix description}
Consider the density matrices of the states going through the linear interferometer 
\begin{equation}
\rho=\sum_{i=1}^{n}p_i \ket{\psi}\bra{\psi},
\end{equation} 
where we have a statistical mixture so that we can find the pure state $\ket{\psi}$ with a probability $p_i$. These matrices are Hermitian and we can always define a skew-Hermitian matrix $v=i\rho \in \mathfrak{u}(M)$ which, in general, will have components both in the tangent and the perpendicular subalgebras.

The evolution of the state represented by the density matrix $\rho$ under a unitary $U \in U(M)$ is
\begin{equation}
\label{AdjEvo}
\rho \rightarrow U\rho U^\dag.
\end{equation}
Introducing an $i$ factor, we can see the output state $i\rho_{\text{out}}=iU\rho_{\text{in}}U^\dag$ corresponds to the effect of the adjoint map on the input. As any operation $U\in  \mathrm{im} \varphi_{m,M}$ must preserve $I_t$ and $I_p$ we have a necessary condition for the transition $\rho_{\text{in}} \rightarrow \rho_{\text{out}}$ to be possible. 

\begin{ftheo}
If an input density matrix $\rho_{\text{in}}$ can be taken to $\rho_{\text{out}}$ using a linear optical system then:
\begin{equation}
I_t(i\rho_{\text{in}})=I_t(i\rho_{\text{out}})
\end{equation}
and
\begin{equation}
I_p(i\rho_{\text{in}})=I_p(i\rho_{\text{out}})
\end{equation}
for the invariants defined in Eqs. (\ref{It}-~\ref{Ip}) which result from taking the sum of the squared projection coefficients of the density matrices into the tangent (image algebra) and perpendicular subspaces.
\end{ftheo}

This is a general result which includes mixed input states, which are also described as Hermitian matrices and evolve following Eq. (\ref{AdjEvo}). If the input state is a pure state, then the norm of the corresponding matrix is 1, as $\langle i\rho, i\rho \rangle =\tr(\rho^2)=1$.

On the other hand
\begin{equation}
\langle i\rho, i\rho \rangle=\left \langle \left(\sum_{i}^{m^2} \langle i\rho , c_i^t\rangle c_i^t+\sum_{j}^{M^2-m^2} \langle i\rho,c_j^p \rangle c_j^p\right),  \left(\sum_{i}^{m^2} \langle i\rho , c_i^t\rangle c_i^t+\sum_{j}^{M^2-m^2} \langle i\rho,c_j^p \rangle c_j^p\right) \right\rangle.
\end{equation}
The inner product is linear and the cross products $\langle c_i^t, c_j^p \rangle=0$ for any valid pair of $i$ and $j$. In our orthonormal bases $\langle c_i^t, c_j^t \rangle$ and $\langle c_i^p, c_j^p \rangle$ are $0$ for any $i\neq j$ and $1$ for $i=j$. This means 
\begin{equation}
\langle i\rho, i\rho \rangle=\sum_{i}^{m^2} |\langle i\rho , c_i^t\rangle |^2 +\sum_{j}^{M^2-m^2} |\langle i\rho,c_j^p \rangle |^2=I_t(\rho)+I_p(\rho)
\end{equation}
and $I_p(i\rho)=1-I_t(i\rho)$ for pure states. For mixed states
\begin{equation}
\label{mixedsum}
    I_t(i\rho)+I_p(i\rho)=\tr(\rho^2)<1.
\end{equation}
As linear optical evolution does not induce any measurement or loss of coherence, transitions between two mixed states with density matrices $\rho_1$ and $\rho_2$ with $\tr(\rho_1^2)\neq \tr(\rho_2^2)$ are impossible. Apart from the direct physical argument, we can see that, if Eq. \eqref{mixedsum} is satisfied, then the two mixed states cannot preserve both the tangent and the perpendicular invariant. Their sums have to give different values of $\tr(\rho^2)$ meaning at least one of the invariants is different for $\rho_1$ and $\rho_2$.

\subsection{Example for two photons and two modes}
We can see a direct application of this result for the state space described in Section \ref{orthexample}. We consider all the states in the Fock basis $\{ \ket{20}, \ket{11}, \ket{02}\}$ and compute the sum of the coefficients of the projection to each basis state numerically. We used the basis matrix calculation in the Python quantum linear optics simulation package QOptCraft \cite{GGMG23}, which has been expanded to add new custom functions created for this purpose \cite{QOptCraft}.

For the states in the Fock basis, we get the invariants
\begin{equation}
I_t(i\ket{20}\bra{20})= 0.83333, \hspace{2ex} I_t(i\ket{11}\bra{11})=0.33333, \hspace{2ex} I_t(i\ket{02}\bra{02})= 0.83333.\nonumber
\end{equation}

We can see, for instance, that it should be possible to go from $\ket{20}$ to $\ket{02}$ and back, which is, indeed, trivial. If we imagine the modes are different paths for the light, swapping the positions of the first and second paths gives the desired operation. However, we see it is impossible to go exactly from $\ket{11}$ to $\ket{20}$ or $\ket{02}$. Superpositions can be reached. In Hong-Ou-Mandel interference in a balanced beamsplitter, an input $\ket{11}$ becomes $\frac{\ket{20}-\ket{02}}{\sqrt{2}}$. 

As expected, $I_t\left(\frac{i}{2}(\ket{20}-\ket{02})(\bra{20}-\bra{02})\right)=0.33333=I_t(i\ket{11}\bra{11})$.

\subsection{The invariant conservation condition is necessary but not sufficient. Counterexample: NOON states.}\label{counterexample}
While one might be tempted to expect the invariant conservation to be also a sufficient condition, we can easily prove this is not the case. 

Consider the so-called NOON states of the form
\begin{equation}
\ket{NOON_n}=\frac{\ket{n0}+\ket{0n}}{\sqrt{2}}
\end{equation}
with $n$ photons in two modes. These states have applications to quantum metrology \cite{LKD02, GLM11}, quantum sensing \cite{DRC17}, to the study of Bell-type inequalities \cite{WLD07} and in precision interferometry \cite{BIW96}, to name a few. 

We start from the no-go result in \cite{vLU07}: if we take from a separable input state with $N$ modes, including ancillas, we cannot obtain a NOON state using only linear optics if any mode has a number of photons greater than the total number of measured photons in the ancillary modes $M$ plus 1. If we have no ancillary modes $N=m=2$ and $M=0$ and then it is only possible to generate NOON states from input states with no more than one photon in any given mode. If $n>2$ there is no input with less than two photons and $\ket{NOON_n}$ cannot be produced from any input state in the Fock basis.

However, numerical tests for inputs $\ket{\frac{n}{2}\frac{n}{2}}$ and a desired output $\ket{NOON_n}$ return the same values for the tangent invariant. For instance, our numerical calculations give the same invariant for the separable inputs and the corresponding NOON states \cite{QOptCraft}:
\begin{align*}
I_t(\ket{11}\bra{11})&=I_t(\ket{NOON_2}\bra{NOON_2})=0.33333, \nonumber\\
I_t(\ket{22}\bra{22})&=I_t(\ket{NOON_4}\bra{NOON_4})=0.20000, \nonumber\\
I_t(\ket{33}\bra{33})&=I_t(\ket{NOON_6}\bra{NOON_6})=0.14286, \nonumber\\
I_t(\ket{44}\bra{44})&=I_t(\ket{NOON_8}\bra{NOON_8})=0.11111, \nonumber\\
I_t(\ket{55}\bra{55})&=I_t(\ket{NOON_{10}}\bra{NOON_{10}})= 0.09090, \nonumber\\
\end{align*}
but only the first transformation is possible (which is basically the Hong-Ou-Mandel interference from the previous section and a phase shift).

\section{Explicit formulae for the tangent invariant of a pure state and consequences}
\label{formulae}
The value of the tangent invariant of a given state can be computed explicitly from a closed formula.

\begin{theorem}\label{teo:tria}
    Given a pure state $\ket{\Psi} \in \mathcal{H}_{m,n}$, the tangent invariant of the associated density matrix is 
    $$
\begin{aligned}
    I_t(i\rho_\Psi)=&C_1(m,n)+C_2(m,n)\sum_{i< j} \left(\langle\hat{a}^\dag_i\hat{a}_j\rangle_\Psi\langle\hat{a}^\dag_j\hat{a}_i\rangle_\Psi-\langle \hat{n}_i\rangle_\Psi\langle \hat{n}_j\rangle_\Psi\right),
\end{aligned}
    $$
    with 
    $$
\begin{aligned}
    C_1(m,n)=&\frac{(m n+1) \Gamma (m+1) \Gamma (n+1)}{\Gamma (m+n+1)},\\
    C_2(m,n)=&\frac{2 \Gamma (m+2) \Gamma (n)}{\Gamma (m+n+1)}.\\
\end{aligned}
    $$
\end{theorem}

The proof is given in the Appendix \ref{proofexplicit}. For states which have a succinct description in the Fock (number state) basis, this approach can be more efficient than doing the computations with the algebra basis. We examine a few immediate results in the next sections.

\subsection{Reduced criterion for allowed transformations}
We can use Theorem \ref{teo:tria} to give a simpler criterion for allowed transformations.

\begin{corollary}\label{corollary_3}
    Let $U\in \mathrm{im} \varphi_{m,M}$. Suppose that 
    $$
\ket{\Psi'}=U\ket{\Psi}.
    $$
    Then
\begin{equation}\label{eq:corollary_3}
\sum_{i< j} \left(\langle\hat{a}^\dag_i\hat{a}_j\rangle_{\Psi'}\langle\hat{a}^\dag_j\hat{a}_i\rangle_{\Psi'}-\langle \hat{n}_i\rangle_{\Psi'}\langle \hat{n}_j\rangle_{\Psi'}\right)=\sum_{i< j} \left(\langle\hat{a}^\dag_i\hat{a}_j\rangle_\Psi\langle\hat{a}^\dag_j\hat{a}_i\rangle_\Psi-\langle \hat{n}_i\rangle_\Psi\langle \hat{n}_j\rangle_\Psi\right).
\end{equation}
\end{corollary}
Here we have the sums of the expected values of pairs of creation and destruction operators and photon number operators 
\begin{equation}
\langle \hat{n}_i\rangle_\Psi=\bra{\Psi}\hat{n}_i\ket{\Psi}
\text{ and }\langle \hat{a}^\dag_i\hat{a}_j\rangle_\Psi=\bra{\Psi}\hat{a}^\dag_i\hat{a}_j\ket{\Psi}.
\end{equation}

Even if we have a superposition of many photon number states, the only terms that are not zero are the terms of the expected values of the photon number operator for each mode (counting the photons in those modes) and the interaction of states which are ``one photon away'' from each other. In $\bra{\Psi}\hat{a}^\dag_i\hat{a}_j\ket{\Psi}$ we get $\langle\Psi|\Phi\rangle$ for a transformed state $\ket{\Phi}$ where we have moved one photon from mode $j$ to mode $i$. If, after moving one photon, none of the states in $\ket{\Phi}$ goes back to $\ket{\Psi}$, the expected value will be zero.  This simplifies exact calculations for many states of interest, as we will see in the following sections.

\section{Transitions from \texorpdfstring{$n$}{n} photons in a single mode to \texorpdfstring{$n - k$}{n-k} and \texorpdfstring{$k$}{k} photons in two different modes are forbidden}\label{ntonminusk}
We can use Corollary \ref{corollary_3} to prove a few general no-go results for linear optics. In this section we show that transformation between photon number states from the Fock basis are, in general, forbidden if the total number of photons $n$ is not equally divided into the $m$ modes (the sets with the number of photons in the occupied modes that sum to $n$ are different).

\begin{corollary}
A Fock state input with $n$ photons in just one of a total of $m\geq 2$ modes cannot be exactly transformed into an output state with $n-k$ photons in one mode and $k$ photons in a second mode using only linear optics for any $k>0$.
\end{corollary}

\begin{proof}
Corollary \ref{corollary_3} implies that, for states in the photon number basis, we only need to worry about the sum 
\begin{equation}\label{reduced}
    \sum_{i< j} \langle \hat{n}_i\rangle_\Psi\langle \hat{n}_j\rangle_\Psi=\sum_{i< j}  n_i n_j.
\end{equation}

Without any loss of generality, we can assume we are transforming an input $\ket{n0\ldots 0}$ into $\ket{(n-k) k 0 \ldots 0}$. Imagine a linear optical system where the modes correspond to different physical paths for the photons. We can always reorder the mode number by switching the paths. This corresponds to a linear interferometer with a scattering matrix $S$ that is a permutation matrix, which is always possible. This is true in general. For any state transition we just need to look into one of the possible permutations of each state.

From Eq. \eqref{reduced} we can see the sum is $0$ for $\ket{n0\ldots 0}$ and $(n-k)k$ for $\ket{(n-k) k 0 \ldots 0}$; the latter is 0 only if $k=0$ or $n=k$ (the trivial cases where we have the original state or a permutation).
\end{proof}
\medskip

The result can be easily generalized to other Fock states. Consider all the restricted partitions of the photon number $n$ up to $m$ terms, i.e. all the ways to write $n$ as a sum of, at most, $m$ integers. We can imagine that each partition into $k$ elements corresponds to the state $\ket{n_1\cdots n_k 0\cdots 0}$. Eq. \eqref{reduced} tells us that most of the exact transformation between states in the Fock basis corresponding to different partitions (and their permuted versions) are forbidden. 

For instance, the state with $n$ photons in one mode has a 0 sum in Eq. \eqref{reduced}. Any other state from the photon number basis has at least two occupied modes and a sum greater than 0. This means the $\ket{n0\cdots0}$ state cannot be transformed deterministically into a single state with a different photon distribution.

Most of the transitions for states corresponding to different partitions can be shown to be impossible with numerical experiments~\cite{QOptCraft}. While there are some transformations that would satisfy Eq.~\eqref{reduced}, like going from $\ket{330}$ to $\ket{114}$ with terms that sum to 9, it is not clear they are not a result of the symmetry in the definition for the invariant, like in the counterexample in  Section \ref{counterexample} for NOON states. We conjecture that, in the given example, the transformation is indeed impossible, but it remains an open problem to study the behaviour as the number of photons and modes grows.

In general, the result is in line with the expected bunching and antibunching behaviour of photons in beam splitters: photons from Fock states tend to stick together in the outputs, unless we admit superpositions of different Fock states, and finding systems that separate photons is not trivial.

 \section{Impossibility of Bell state generation from separable photon number states with photon number ancillary states}\label{bell}

Bell states are a crucial resource in quantum information. They are a key element in Ekert's quantum key distribution \cite{Eke91}, quantum teleportation \cite{BBC93}, entanglement swapping \cite{PBW98}, superdense coding \cite{BW92} or quantum secure direct communication \cite{DLL03,SZL22}. Also, they can be used in universal quantum computation with linear optics \cite{BR05} and are a fundamental ingredient in quantum networks \cite{PBW98, GZG22} or distributed secure quantum machine learning \cite{SZ17}.

While they can be generated with nonlinear processes at moderate rates \cite{KMW95,BCG07}, there is a strong interest in creating entangled Bell states with linear optics acting on simple input states. There are many different probabilistic proposals \cite{ZPM02,ZBL08,FSK21} conditioned on certain measurement results. Most proposals include a herald: a recognizable ancillary state that, when measured in part of the output modes, guarantees the unmeasured modes contain the desired Bell state. These schemes have a limited probability of success, which is given by the probability of finding the herald state in the measurement. We can wonder whether this is a limitation of the resources and if using larger state spaces with more modes and ancillary photons would allow to generate a Bell state with certainty from simple input states.  We give a negative answer for a large family of heralded schemes which expand previous impossibility proofs for systems with two and three photons \cite{SLM17}.

If we consider two-level systems, or qubits, with basis states $\ket{0}$ and $\ket{1}$, a Bell state is defined as $$\frac{\ket{00}+\ket{11}}{\sqrt{2}}.$$ This is an entangled state where a measurement in the given computational basis will always produce the same result on both qubits. We consider the analogue in linear optics for dual-rail encoding with a photon number description $\ket{0}\equiv\ket{01}$ and $\ket{1}\equiv\ket{10}$. The desired Bell state is 
$$
\frac{\ket{0101}+\ket{1010}}{\sqrt{2}}
$$ 
where we have four orthogonal photonic modes. The most usual physical realization of this state is two separate paths with two orthogonal polarizations having a state of the form 
$$
\frac{\ket{H}_1\ket{H}_2+\ket{V}_1\ket{V}_2}{\sqrt{2}}
$$ 
for horizontally and vertically polarized photons, but the description in terms of four modes also covers more general situations.

We cannot produce this output state as a subsystem of the output modes even if we use any number of ancillary photons and modes when the input is a state from the Fock basis. The only assumption is that the output in the measured ancillary modes must also be a state in the Fock basis. 

\begin{corollary}
An input Fock state $\ket{\psi} = \ket{n_1\, \ldots\, n_m}$ cannot be deterministically transformed with linear optics into an output Bell state with an ancilla
$$\ket{\psi'} = \frac{1}{\sqrt{2}}\big(\ket{1010}+\ket{0101}\big)\ket{\text{aux}\,'} \:,$$
where $\ket{\text{aux}\,'} = \ket{n_5'\, n_6'\, \ldots\, n_{m}'}$ is the ancilla.

\end{corollary}
\begin{proof}
    This result is proved after the application of formula \eqref{eq:corollary_3} in Corollary \ref{corollary_3}. First, note that $\langle \hat{a}^\dagger_i \hat{a}_j \rangle = 0$\,, since the two Fock states in $\ket{\psi'}$ are more than one photon away from each other. Thus, we only need to compute the terms containing $\langle \hat{n}_i\rangle$:
    \begin{alignat*}{2}
        & \langle \hat{n}_i\rangle_\psi = n_i \:, \qquad && i=1, \ldots m \:;\\[2mm]
        & \langle \hat{n}_i\rangle_{\psi'} = \frac{1}{2} \:, && i=1, \ldots 4 \:;\\[2mm]
        & \langle \hat{n}_i\rangle_{\psi'} = n_i' \:, \qquad && i=5, \ldots m \:.
    \end{alignat*}
    
    To prove the impossibility, we only need to make sure that the difference between the invariants is non-zero. The sum $\sum_{i<j}\langle \hat{n}_i\rangle\langle \hat{n}_j\rangle$ has an integer value for $\ket{\psi}$. It is also an integer for $\ket{\psi'}$ when $i, j=5,\ldots m$. The only terms with, possibly, non-integer terms are those involving $\langle \hat{n}_i\rangle = 1/2$\,. First,
    \begin{equation*}
        \sum_{\substack{i=1\ldots 4\\[1mm]j=5\ldots m}}\langle \hat{n}_i\rangle_{\psi'}\langle \hat{n}_j\rangle_{\psi'} = \frac{1}{2}\,4\sum_{j=5\ldots m} n_j'
    \end{equation*}
    turns out to be also an integer. Only the term independent of the auxiliary state is to make a crucial difference:
        \begin{equation*}
        \sum_{\substack{i<j\\[1mm]i,j=1\ldots 4}}\langle \hat{n}_i\rangle_{\psi'}\langle \hat{n}_j\rangle_{\psi'} = \frac{1}{2}\frac{1}{2}\frac{4(4-1)}{2}= \frac{3}{2}\:;
    \end{equation*}
    which means that the invariants will always be, at least, $1/2$ apart from each other. We conclude that the transition is, indeed, impossible.
    
\end{proof}

While this is a strong negative result there are two important caveats. First, we require the herald to be a fixed state from the Fock basis. There could be a superposition of Fock states in the ancilla that guarantees a deterministic Bell state in the first four modes. However, the proof also excludes any of these superpositions that can be generated using linear optics and any single Fock state, as they can be converted back into Fock states with the inverse transformation in the ancillary modes. These states can always be converted into the states covered in the proof with the inverse linear system in the ancillary modes and, if we can reach them, then the original fixed state may be also attained. Second, the proof only restricts deterministic generation of Bell states. It is still open if there is an upper limit on the maximum probability of success in heralded schemes that add a larger number of photons or modes. We leave it for future investigation whether the relative differences between two invariants can be related to this probability of success.

\section{Impossibility of exact transformation between dual-rail GHZ states and W states with photon number ancillas}\label{GHZtoW}
The same technique can be used to prove the impossibility of transitioning from a state $\ket{\text{GHZ}} = (\ket{000} + \ket{111})/\sqrt{2}$ encoded in dual-rail to a state $\ket{\text{W}} = (\ket{100} + \ket{010} + \ket{010})/\sqrt{3}$ also in dual-rail encoding. These states are already known to be entangled in inequivalent ways, meaning that one cannot use local operations and classical communication to transform one into the other \cite{DVC00}.

\begin{corollary}
A GHZ state with a state number ancilla,
$$\ket{\text{GHZ}}\ket{aux} = \frac{\ket{010101} + \ket{101010}}{\sqrt{2}}\,\ket{n_7\ldots n_m}\:,$$
cannot be deterministically transformed with linear optics into a W state with any state number ancilla,
$$\ket{\text{W}}\ket{aux'} = \frac{\ket{100101} + \ket{011001} + \ket{010110}}{\sqrt{3}} \, \ket{n_7'\ldots n_m'}\:.$$

\end{corollary}
\begin{proof}
    In a similar fashion to the previous proof, we show that the terms in the invariant arising from $\ket{\text{GHZ}}$ and $\ket{\text{W}}$ give a different non-integer value, while the remaining terms are all integers: 
    \begin{alignat*}{2}
        & \langle \hat{n}_i\rangle_{\text{GHZ}} = \frac{1}{2} \:, \qquad && i=1, \ldots 6 \:;\\[2mm]
        & \langle \hat{n}_i\rangle_{\text{W}} = \frac{1}{3} \:, \qquad && i=1,3,5 \:;\\[2mm]
        & \langle \hat{n}_i\rangle_{\text{W}} = \frac{2}{3} \:, \qquad && i=2,4,6\:.
    \end{alignat*}
    This yields a different invariant for the first 6 modes:
    \begin{equation*}
        \sum_{i<j}\langle \hat{n}_i\rangle_{\text{GHZ}} \, \langle \hat{n}_j\rangle_{\text{GHZ}} = \frac{15}{4} \neq \sum_{i<j}\langle \hat{n}_i\rangle_{\text{W}} \, \langle \hat{n}_j\rangle_{\text{W}} = \frac{11}{3}\:.
    \end{equation*}
    The remaining terms involving the auxiliary photons result in integers; thus, the invariants can never be equal.
    
\end{proof}

This result is consistent with the intuition that converting between different families of entangled states requires entangling gates, which would allow universal quantum computation with linear optics, which is impossible without the assistance of measurement and postcorrection or nonlinear systems.
  
\section{Closing remarks}\label{conclusions}
We have described a necessary criterion that, for any two given input states, identifies when it is impossible to transform one into the other using only linear optical devices. 

For states which are a superposition of few photon number states, the explicit formulas can be efficiently computed, even for large state spaces. For general density matrices there are numerical methods based on the projections on well-defined subspaces.

Both methods have been added as functions in the open software package QOptCraft \cite{GGMG23} and we have made available some self-guided example demonstrations \cite{QOptCraft}.

The closed formulas can be used to give general theoretical no-go results for quantum linear optics that include the impossibility of most transitions between photon number states with a different photon distribution between the possible modes, the impossibility of the exact heralded generation of Bell states from separable Fock inputs and the impossibility of exactly transforming between GHZ and W states with fixed ancillary photon number states.

These results have applications to optical quantum state generation. Advanced entangled optical states are a basic resource in optical quantum computation and quantum communication protocols. Having a formal description of the limitations can help in the search for optical systems producing interesting output states. 

Finally, the software for numerical computation can also help in the automated search for a target output. Computer assisted methods can find useful quantum states of light and there is a growing field of research working with linear optics \cite{Kno16, MNK18,NMR19,  ONK19,BHV20,GEZ20,WMD20,AEK21,KKT21}. The given transformation criterion can be useful as a screening method that can quickly weed out impossible inputs without the need to work with the whole state space.

\section{Authorship contribution statement}
Pablo V. Parellada developed the software used for the article. All four authors contributed equally to the rest of the work.

\section{Declaration of Competing Interest}
The authors declare that they have no known competing financial interests or personal relationships that could have appeared to influence the work reported in this paper.

\section{Data availability}\label{sec:data}

The source code of the library can be found at \url{https://gitlab.tel.uva.es/juagar/qoptcraft}. All the computations for this article are available in a jupyter notebook
\url{https://gitlab.tel.uva.es/juagar/qoptcraft/-/blob/main/examples/article.ipynb}.
The latest version of the package can be cited as \cite{QOptCraft}.

\section{Acknowledgments}
P. V. Parellada has been funded by the European Union NextGenerationEU (PRTRC17.I1) and the Consejer\'ia de Educaci\'on, Junta de Castilla y Le\'on, through QCAYLE project. V. Gimeno i Garcia has been partially supported by the Research Program of the University Jaume I--Project UJI-B2018-35, as well as by the Spanish Government and FEDER grants PID2020-115930GA-I00 (MICINN) and  AICO/2021/252. Concerning J.J. Moyano-Fern\'andez, this work was supported in part by Grant TED2021-130358B-I00 funded by MCIN/AEI/10.13039/501100011033 and by the “European Union 
NextGenerationEU/PRTR”, as well as by Universitat Jaume I, grant GACUJIMA/2023/06. J.C. Garcia-Escartin has been funded by the the European Union NextGenerationEU (PRTRC17.I1) and the Consejer\'ia de Educaci\'on, Junta de Castilla y Le\'on, through QCAYLE project 
and the Spanish Government and FEDER grant PID2020-119418GB-I00 (MICINN).

\section{Note added}
After acceptance, it has been brought to our attention the paper “Multiphoton states related via linear optics” from Piotr Migdał et al.~\cite{MRO14}, that studies the same problem and finds a series of conserved quantities in linear optical transformations using a particle description as well as the mode description we employ.

\section*{Appendix A. Proof of Theorem \ref{teo:tria}}
\label{proofexplicit}
To prove Theorem \ref{teo:tria}, first of all we need to show the following proposition:
\begin{proposition}
Consider the following basis for $\dd$,
    \begin{align}
\label{ImBasis2}
b_j:=&b_{j}^{n}=i\hat{n}_j \mbox{ for } \, j=1,\ldots , m; \nonumber \\
c_{jk}:=&b_{jk}^{e}=\frac{i}{2}\big (\hat{a}^\dag_j \hat{a}_k+\hat{a}^\dag_k \hat{a}_j\big ) \mbox{ for } \, k< j=1,\ldots , m; \nonumber  \\
d_{jk}:=&b_{jk}^{f}=\frac{1}{2}\big (\hat{a}^\dag_j \hat{a}_k-\hat{a}^\dag_k \hat{a}_j\big ) \mbox{ for } \, k< j=1,\ldots , m,
\end{align}
as well as the inner product 
\begin{equation}\label{eqn:skalar}
\langle u,v \rangle:=\frac{1}{2}\mathrm{tr}(u^{\dag} v + v^{\dag}u)
\end{equation}
for any $u,v\in \mathfrak{u}(M)$. We have that
\begin{equation}
\begin{aligned}
    \langle b_i,b_j\rangle =& \left\lbrace\begin{array}{ccc}
  A & \mathrm{if}& i=j \\
  B  & \mathrm{if}& i\neq j
\end{array}\right.\\
\langle b_k,c_{ij}\rangle=&\langle b_k,d_{lm}\rangle=\langle c_{ij},d_{lm}\rangle=0\\
\langle c_{ij},c_{lm}\rangle&=\langle d_{ij},d_{lm}\rangle=\frac{1}{2}\delta_{il}\delta_{jm}\left(B+C\right),
\end{aligned}
\end{equation}
where
$$
\begin{aligned}
    A:=&\frac{(m+2 n-1) \Gamma (m+n)}{\Gamma (m+2) \Gamma (n)},\\
    B:=&\frac{\Gamma (m+n)}{\Gamma (m+2) \Gamma (n-1)},\\
    C:=&\frac{\Gamma (m+n)}{\Gamma (m+1) \Gamma (n)}.\\
\end{aligned}
$$
\end{proposition}
\begin{proof}
Observe that the given basis is not orthonormal for the case of $n>1$. Indeed, 
\begin{equation}\label{eq:product}
\begin{aligned}
\langle b_i,b_i\rangle=&\mathrm{tr}\left(\hat{n}_i\hat{n}_j\right)=\sum_{n_1+n_2+\cdots+n_m=n} \langle n_1,\cdots,n_m\vert \hat{n}_i\hat{n}_j\vert n_1,\cdots,n_m\rangle,    \\
\langle b_i,c_{jk}\rangle=&\mathrm{tr}\left(\hat{n}_ic_{jk}\right) =   \sum_{n_1+n_2+\cdots+n_m=n} \langle n_1,\cdots,n_m\vert \hat{n}_i\hat{c}_{jk}\vert n_1,\cdots,n_m\rangle=0,\\
\langle b_i,d_{jk}\rangle=&\mathrm{tr}\left(\hat{n}_id_{jk}\right) =   \sum_{n_1+n_2+\cdots+n_m=n} \langle n_1,\cdots,n_m\vert \hat{n}_i\hat{d}_{jk}\vert n_1,\cdots,n_m\rangle=0,\\
\langle c_{ij},c_{lm}\rangle=&\frac{1}{4}\mathrm{tr}\left((\hat{a}^\dag_i \hat{a}_j+\hat{a}^\dag_j \hat{a}_i\big )(\hat{a}^\dag_l \hat{a}_m+\hat{a}^\dag_m \hat{a}_l\big )\right)=\frac{\delta_{im}\delta_{lj}+\delta_{il}\delta_{jm}+\delta_{jm}\delta_{il}+\delta_{jl}\delta_{im}}{4}\mathrm{tr}\left(\hat{n}_m+\hat{n}_m\hat{n}_l\right)\\
=&\frac{1}{2}\delta_{il}\delta_{jm}\mathrm{tr}\left(\hat{n}_m+\hat{n}_m\hat{n}_l\right)=\frac{1}{2}\delta_{il}\delta_{jm}\sum_{n_1+n_2+\cdots+n_m=n} \langle n_1,\cdots,n_m\vert \hat{n}_m+\hat{n}_m\hat{n}_l\vert n_1,\cdots,n_m\rangle,\\
\langle c_{ij},d_{lm}\rangle=&\frac{-i}{4}\mathrm{tr}\left((\hat{a}^\dag_i \hat{a}_j+\hat{a}^\dag_j \hat{a}_i\big )(\hat{a}^\dag_l \hat{a}_m-\hat{a}^\dag_m \hat{a}_l\big )\right)=0,\\
\langle d_{ij},d_{lm}\rangle=&\frac{1}{2}\delta_{il}\delta_{jm}\mathrm{tr}\left(\hat{n}_m+\hat{n}_m\hat{n}_l\right)=\frac{1}{2}\delta_{il}\delta_{jm}\sum_{n_1+n_2+\cdots+n_m=n} \langle n_1,\cdots,n_m\vert \hat{n}_m+\hat{n}_m\hat{n}_l\vert n_1,\cdots,n_m\rangle.
\end{aligned}
\end{equation}
Because of the independence in the particular mode, and assuming that there are at least two modes, we have
$$
\begin{aligned}
&\sum_{n_1+n_2+\cdots+n_m=n} \langle n_1,\cdots,n_m\vert \hat{n}_i\hat{n}_j\vert n_1,\cdots,n_m\rangle=\left\lbrace\begin{array}{ccc}
  A & \mathrm{if}& i=j \\
  B  & \mathrm{if}& i\neq j
\end{array}\right.\\
&\sum_{n_1+n_2+\cdots+n_m=n} \langle n_1,\cdots,n_m\vert \hat{n}_m+\hat{n}_m\hat{n}_l\vert n_1,\cdots,n_m\rangle=C+B
\end{aligned}
$$
with 
\begin{equation}\label{eq:36}
\begin{aligned}
A:=&\displaystyle \sum_{n_1+n_2+\cdots+n_m=n} \langle n_1,\cdots,n_m\vert \hat{n}_1^2\vert n_1,\cdots,n_m\rangle,\\
B:=& \displaystyle\sum_{n_1+n_2+\cdots+n_m=n}\langle n_1,\cdots,n_m\vert \hat{n}_1\hat{n}_2\vert n_1,\cdots,n_m\rangle,\\
C:=&\displaystyle \sum_{n_1+n_2+\cdots+n_m=n} \langle n_1,\cdots,n_m\vert \hat{n}_1\vert n_1,\cdots,n_m\rangle.    
\end{aligned}    
\end{equation}
Therefore, 
\begin{equation}\label{productebase}
\begin{aligned}
    \langle b_i,b_j\rangle =& \left\lbrace\begin{array}{ccc}
  A & \mathrm{if}& i=j \\
  B  & \mathrm{if}& i\neq j
\end{array}\right.\\
\langle b_k,c_{ij}\rangle=&\langle b_k,d_{lm}\rangle=\langle c_{ij},d_{lm}\rangle=0\\
\langle c_{ij},c_{lm}\rangle&=\langle d_{ij},d_{lm}\rangle=\frac{1}{2}\delta_{il}\delta_{jm}\left(B+C\right).
\end{aligned}
\end{equation}
To finish the proof of the proposition we only have to calculate the values of $A,B,C$ given in \eqref{eq:36}. An easy computation using \textsf{Mathematica} \cite{Mathematica} shows that
$$
\begin{aligned}
    A=&\sum_{k=0}^n(n-k)^2\binom{m+k-2}{k}=\frac{(m+2 n-1) \Gamma (m+n)}{\Gamma (m+2) \Gamma (n)},\\
    C=&\sum_{k=0}^n(n-k)\binom{m+k-2}{k}=\frac{\Gamma (m+n)}{\Gamma (m+1) \Gamma (n)}.
\end{aligned}
$$
Moreover,
$$
\begin{aligned}
    Mn^2=&\sum_{n_1+n_2+\cdots+n_m=n} \langle n_1,\cdots,n_m\vert \sum_{i=1}^m\hat{n}_i\sum_{j=1}^m\hat{n}_j\vert n_1,\cdots,n_m\rangle\\
    =&\sum_{i=1}^m\sum_{n_1+n_2+\cdots+n_m=n}\langle n_1,\cdots,n_m\vert \hat{n}_i^2\vert n_1,\cdots,n_m\rangle\\
    &+\sum_{i\neq j}^m\sum_{n_1+n_2+\cdots+n_m=n}\langle n_1,\cdots,n_m\vert \hat{n}_i\hat{n}_j\vert n_1,\cdots,n_m\rangle\\
    =&mA+m(m-1)B
\end{aligned}
$$
and finally,
$$
\begin{aligned}
    B=&\frac{1}{m(m-1)}\left(Mn^2-mA\right)=\frac{\Gamma (m+n)}{\Gamma (m+2) \Gamma (n-1)}.
\end{aligned}
$$
\end{proof}

\medskip

Once the above proposition is proved we are ready to show Theorem \ref{teo:tria}. Given a pure state $\ket{\Psi}$ consider the associated density matrix 
$
\rho_\Psi=\ket{\Psi}\bra{\Psi}
$, then $i\rho_\Psi$ belongs to $\mathfrak{u}(M)$ and has a tangent component $v_t\in\dd$.  In our basis therefore,
$$
v_t=v_t^i b_i+v_t^{lm}c_{lm}+v_t^{\alpha \beta}d_{\alpha \beta},
$$
where we use the Einstein summation convention. Hence
\begin{equation}\label{eq:19}
\begin{aligned}
I_t(i\rho_\Psi)=&\langle v_t, v_t\rangle=\langle i\rho_\Psi, v_t\rangle=\langle i\rho_\Psi, v_t^i b_i+v_t^{lm}c_{lm}+v_t^{\alpha\beta}d_{\alpha\beta}\rangle\\
=&v_t^i\langle i\rho_\Psi,  b_i\rangle+v_t^{lm}\langle i\rho_\Psi, c_{lm}\rangle+v_t^{\alpha \beta}\langle i\rho_\Psi,d_{\alpha\beta}\rangle\\
=&v_t^i\langle ib_i\rangle_\Psi+v_t^{lm}\langle ic_{lm}\rangle_\Psi+v_t^{\alpha\beta}\langle id_{\alpha\beta}\rangle_\Psi,
\end{aligned}    
\end{equation}
where 
$$
\langle ib_i\rangle_\Psi=\bra{\Psi} ib_i\ket{\Psi},\quad \langle ic_{lm}\rangle_\Psi=\bra{\Psi} ic_{lm}\ket{\Psi}, \quad \mathrm{and}\quad \langle id_{ro}\rangle_\Psi=\bra{\Psi} id_{ro}\ket{\Psi}.
$$
On the other hand, using \eqref{productebase} we have
$$
\begin{aligned}
\langle ib_i\rangle_\Psi=\langle b_i, i\rho_\Psi\rangle=\langle b_i, v_t\rangle=v^j_t\langle b_i,b_j\rangle=g_{ij} v^j_t
\end{aligned}
$$
where $g$ is the matrix with elements $g_{ij}=\langle b_i,b_j\rangle$.
Then, 
$$
v^j_t=g^{ji}\langle ib_i\rangle_\Psi,
$$
with $g^{ij}=\left(g^{-1}\right)_{ij}$. After an algebraic manipulation we obtain
\begin{equation}
    \begin{aligned}
        \left(g^{-1}\right)_{ij}=\left\lbrace\begin{array}{ccc}
  \frac{m-2+\alpha}{B(\alpha-1)(m-1+\alpha)} & \mathrm{if}& i=j \\
  \frac{-1}{B(\alpha-1)(m-1+\alpha)}  & \mathrm{if}& i\neq j
\end{array}\right.
    \end{aligned}
\end{equation}
with $\alpha=A/B$.
Analogously,
$$
\begin{aligned}
\langle ic_{ij}\rangle_\Psi=\langle c_{ij}, i\rho_\Psi\rangle=\langle c_{ij}, v_t\rangle=v^{lm}_t\langle c_{ij},c_{lm}\rangle=\frac{1}{2}(C+B)v^{ij}_t.
\end{aligned}
$$
Therefore
\begin{equation}
    v^{ij}_t=\frac{2}{B+C}\langle ic_{ij}\rangle_\Psi,
\end{equation}
and similarly,
\begin{equation}
  v^{\alpha \beta}_t=\frac{2}{B+C}\langle id_{\alpha \beta}\rangle_\Psi.   
\end{equation}
Then \eqref{eq:19} can be rewritten as 
\begin{equation}
\begin{aligned}
    I_t(i\rho_\Psi)=&\frac{m-2+\alpha}{B(\alpha-1)(m-1+\alpha)}\sum_{i=1}^m\langle \hat{n}_i\rangle_\Psi^2-\frac{1}{B(\alpha-1)(m-1+\alpha)}\sum_{i\neq j} \langle \hat{n}_i\rangle_\Psi\langle \hat{n}_j\rangle_\Psi\\
    &+\sum_{i<j}\frac{2}{B+C}\langle ic_{ij}\rangle_\Psi^2+\sum_{i<j}\frac{2}{B+C}\langle id_{ij}\rangle_\Psi^2
\end{aligned}
\end{equation}
Using elementary manipulations we have
\begin{equation}
\begin{aligned}
    I_t(i\rho_\Psi)=&\frac{(m-2)B+A}{(A-B)((m-1)B+A)}n^2-\frac{2}{A-B}\sum_{i< j} \langle \hat{n}_i\rangle_\Psi\langle \hat{n}_j\rangle_\Psi\\
    &+\frac{2}{B+C}\sum_{i<j}\langle\hat{a}^\dag_i\hat{a}_j\rangle_\Psi\langle\hat{a}^\dag_j\hat{a}_i\rangle_\Psi,
\end{aligned}
\end{equation}
but then
$$
\begin{aligned}
    C_1(m,n):=&\frac{(m-2)B+A}{(A-B)((m-1)B+A)}n^2=\frac{(m n+1) \Gamma (m+1) \Gamma (n+1)}{\Gamma (m+n+1)},\\
    C_2(m,n):=&\frac{2}{A-B}=\frac{2 \Gamma (m+2) \Gamma (n)}{\Gamma (m+n+1)}=\frac{2}{B+C},\\
\end{aligned}
$$
and the theorem is proved.

\bibliography{AllowedAndForbidden}

\end{document}